\documentclass[a4paper,11pt]{article}

\usepackage{eucal}
\usepackage{amsfonts,amsmath,amsthm,amscd}
\usepackage{graphics,epsfig}
\usepackage{mathrsfs,latexsym}
\usepackage{fancyhdr}
\usepackage{cite}

\def\bbc{{\Bbb C}}
\def\bbr{{\Bbb R}}

\def\bbz{{\Bbb Z}}

\def\tr{\mathrm{tr\,}}

\def\re{\mathrm{Re\,}}
\def\im{\mathrm{Im\,}}

\def\rmd{\mathrm{d\,}}
\def\rmi{\mathrm{i}}
\def\rme{\mathrm{e}}

\def\openone{\leavevmode\hbox{\small1\kern-3.3pt\normalsize1}}
\def\diag{\mbox{diag\,}}

\newtheorem{prop}{Proposition}

\newtheorem{theo}{Theorem}
\newtheorem{corol}{Corollary}

\newtheorem{rem}{Remark}

\begin{document}
\title{A Generic Nonlinear Evolution Equation\\ of Magnetic Type II. Particular Solutions}
\author{T. Valchev\\
\small Institute of Mathematics and Informatics\\
\small Bulgarian Academy of Sciences\\
\small Acad. Georgi Bonchev Str., 1113 Sofia, Bulgaria\\
\small E-mail: tiv@math.bas.bg}
\date{}
\maketitle
\begin{abstract}
We consider a matrix nonlinear partial differential equation that generalizes Heisenberg ferromagnet equation.
This generalized Heisenberg ferromagnet equation is completely integrable with a linear bundle Lax pair related
to the pseudo-unitary algebra. This allows us to explicitly derive particular solutions by using dressing technique.
We shall discuss two classes of solutions over constant background: soliton-like solutions and quasi-rational solutions.
Both classes have their analogues in the case of the Heisenberg ferromagnet equation related to the same Lie algebra.
\end{abstract}

\section{Introduction}\label{intro}

A classical integrable model \cite{BoPo90,forkul,book,blue-bible} is given by Heisenberg ferromagnet equation
\begin{equation}
\rmi S_t = \frac{1}{2}[S,S_{xx}] ,\qquad S^2 = I.
\label{hf}\end{equation}
Above, the function $S$ takes values in a simple complex matrix Lie algebra $\mathfrak{g}$, brackets mean commutator,
$I$ stands for the unit matrix, "$\rmi$" is the imaginary unit and the subscripts denote differentiation. In the simplest case when $\mathfrak{g} = \mathfrak{su}(2)$ (\ref{hf}) is known to describe the dynamics of a spin chain in continuous limit.

Various generalizations or analogues to the equation (\ref{hf}) have been proposed in literature, see 
\cite{golsok,ishim,mmnl,pliska2017} for some representative examples. The present paper is dedicated to the following
matrix nonlinear evolution equation (NEE)
\begin{equation}
\rmi S_t = a [S,S_{xx}] + \frac{3a}{2} \left(S^2S_{x}S\right)_x + b \left[S^2, S_{x}\right]_x ,
\qquad a,b\in\bbc,
\label{ghf}	
\end{equation}
that was introduced in \cite{bgsiam} as a natural candidate to generalize (\ref{hf}). Above, $S(x,t)$ belongs to
the Lie algebra $\mathfrak{sl}(n,\bbc)$ and fulfills the constraint
\[[S(x,t)]^3 = S(x,t) .\]
Like the Heisenberg ferromagnet equation, the NEE (\ref{ghf}) admits zero curvature representation with a Lax pair of the form:	
\begin{eqnarray*}
L(\lambda) & = & \rmi\partial_x - \lambda S,\qquad A(\lambda) = \rmi\partial_t + \lambda A_1 + \lambda^2 A_2,\\
A_2 & = &  a S + b\left(S^2 - \frac{2r}{n}I\right),\qquad r = \frac{\tr(S^2)}{2}\, ,\\
A_1 & = & \rmi a [S,S_{x}] + \frac{3\rmi a}{2}S^2S_{x}S + \rmi\, b \left[S^2, S_{x}\right],
\end{eqnarray*}
where $\lambda\in\bbc$ is spectral parameter.

It is not hard to see that the equation (\ref{ghf}) leads to (\ref{hf}) if $S(x,t)$ is a non-singular matrix, see \cite{bgsiam}.
On the other hand, if we impose a $\bbz_2$ reduction in Mikhailov's sense \cite{mikh2} that connects the Lax pair
above to a symmetric space of the series $\mathbf{A.III}$, we obtain the NEE 
\[\rmi S_t =  b \left[S^2, S_{x}\right]_x ,\qquad b\in\bbc , \]
that was studied in a series of papers, e.g. see \cite{side9,pliska2017}.

Our purpose in this paper is to derive special solutions of a pseudo-Hermitian reduction of the generalized Heisenberg
ferromagnet equation (\ref{ghf}), i.e., we shall assume that $S$ fulfills the symmetry condition
\[\mathcal{E}S^{\dag}(x,t)\mathcal{E} = S(x,t),\]
where $\mathcal{E}$ is a diagonal matrix with $1$ or $-1$ on its diagonal and $\dag$ stands for Hermitian conjugation.
As a result of this, the constants $a$ and $b$ appearing in (\ref{ghf}) will be some real numbers. This is the most important case
in terms of potential applications. The method of integration we shall apply is Zakharov-Shabat's dressing method. Dressing
method is equally well suited to deal with NEEs whose Lax pairs are related to simple Lie algebras regardless of their
rank or whether or not some additional algebraic constraints are imposed on the Lax operators. By using dressing factors
that are meromorphic functions with simple poles in the spectral parameter, we shall demonstrate how one can construct two
classes of trivial background solutions: soliton-like solutions and quasi-rational solutions. The soliton solutions correspond
to complex poles in generic position and tend exponentially fast to some constants as $x\to\pm\infty$. The second class consists
of solutions that are obtained by using real pole dressing factors. This introduces certain "degeneracy" in the whole procedure
which leads to solutions that can be represented as the ratio of quasi-polynomials, i.e., the product of polynomials with oscillating exponential functions. Common features of both types of particular solutions are that they are not traveling waves
and may have singularities.

The paper is organized as follows. Next section is preliminary and discusses the general concept and algorithm of
dressing method as applied to linear bundle Lax pairs in pole gauge. Therein we shall present most basic facts to be used
further in text, thus giving more or less a self-contained exposition on the subject. Section \ref{solutions} contains our
main results and it is divided into two subsections. Subsection \ref{solitons} is dedicated to solutions of the soliton type
while Subsection \ref{quasi-rational} contains quasi-rational solutions. Apart of the above-mentioned types of solutions to
the generalized Heisenberg ferromagnet equation (\ref{ghf}), we also pay attention to their counterparts for the special case
of the equation (\ref{hf}). Since the latter is a rather important case by itself, we present the corresponding results in the form
of separate theorems. The last section contains further discussion and our concluding remarks.

\section{General Remarks on Dressing Method}
\label{prelim}

In this section we shall explain the idea that underlies dressing method and demonstrate how it can be applied to the generalized
Heisenberg ferromagnet equation (\ref{ghf}) in order to generate special solutions, see \cite{book, zakh-mikh, ZS} for more detailed explanations. 

Let us consider a Lax pair of the general form 
\begin{eqnarray}
L^{(0)}(\lambda) &=& \rmi\partial_x + U^{(0)}(x,t,\lambda),\label{lax01}\\
A^{(0)}(\lambda) &=& \rmi\partial_t + V^{(0)}(x,t,\lambda),\label{lax02}
\end{eqnarray}
where $U^{(0)}(x,t,\lambda)$ and $V^{(0)}(x,t,\lambda)$ are some polynomials in the spectral parameter $\lambda$
with coefficients in the Lie algebra $\mathfrak{sl}(n,\bbc)$, i.e., we can write down
\begin{eqnarray*}
U^{(0)}(x,t,\lambda) &=& \sum^{N}_{j=0}U^{(0)}_j(x,t) \lambda^j,\qquad U^{(0)}_j(x,t)\in\mathfrak{sl}(n,\bbc),\\
V^{(0)}(x,t,\lambda) &=& \sum^{M}_{k=0}V^{(0)}_k(x,t)\lambda^k,\qquad V^{(0)}_k(x,t)\in\mathfrak{sl}(n,\bbc),
\qquad N\leq M.
\end{eqnarray*}
We assume that $U^{(0)}_j$, $j=0,\ldots, N$ and $V^{(0)}_k$, $k=0,\ldots,M$ are some known smooth functions in $x$ and $t$
defined almost everywhere in $\bbr^2$. We shall call (\ref{lax01}) and (\ref{lax02}) bare Lax operators and shall require
that the linear problems 
\begin{eqnarray}
L^{(0)}(\lambda)\Psi^{(0)}(x,t,\lambda) &=& 0,\label{bareaux1}\\
A^{(0)}(\lambda)\Psi^{(0)}(x,t,\lambda) &=& \Psi^{(0)}(x,t,\lambda)C(\lambda)\label{bareaux2}
\end{eqnarray} 
admit a common matrix solution $\Psi^{(0)}(x,t,\lambda)$, $\det\Psi^{(0)}(x,t,\lambda)\neq 0$ for some non-singular matrix
$C(\lambda)$. Obviously, this is possible if the operators $L^{(0)}(\lambda)$ and $A^{(0)}(\lambda)$ commute.
This compatibility condition of bare linear problems allows one to express the coefficients $V^{(0)}_k$, $k=0,\ldots,M$
through $U^{(0)}_j$, $j=0,\ldots, N$ and leads to some NEE for the latter.

Let us now construct another matrix function
\begin{equation}
\Psi^{(1)}(x,t,\lambda) = \mathcal{G}(x,t,\lambda)\Psi^{(0)}(x,t,\lambda),
\label{dres}\end{equation}
where the multiplier $\mathcal{G}(x,t,\lambda)$, $\det\mathcal{G}(x,t,\lambda) \neq 0$ is called dressing factor. Then the
Lax operators $L^{(0)}(\lambda)$ and $A^{(0)}(\lambda)$ are transformed into
\begin{equation}
\begin{split}
L^{(1)}(\lambda) &= \mathcal{G}(x,t,\lambda)L^{(0)}(\lambda)\left[\mathcal{G}(x,t,\lambda)\right]^{-1},\\
A^{(1)}(\lambda) &= \mathcal{G}(x,t,\lambda)A^{(0)}(\lambda)\left[\mathcal{G}(x,t,\lambda)\right]^{-1}.
\end{split}
\label{laxdres}
\end{equation} 
We recall that the resolvent operator $R^{(0)}(\lambda)$ of the bare scattering operator $L^{(0)}(\lambda)$ is defined through
the equality 
\begin{equation}
L^{(0)}(\lambda)R^{(0)}(\lambda) = I,
\label{bareresolv}
\end{equation}
where $I$ is the unit matrix. Similarly, for the resolvent operator of $L^{(1)}(\lambda)$ we have
\begin{equation}
L^{(1)}(\lambda)R^{(1)}(\lambda) = I.
\label{dresresolv}\end{equation}
In view of (\ref{laxdres}), (\ref{bareresolv}) and (\ref{dresresolv}) we immediately see that the following proposition
holds.
\begin{prop}
The resolvent operators of bare and dressed scattering operators are interrelated through:
\[R^{(1)}(\lambda) = \mathcal{G}(x,t,\lambda)R^{(0)}(\lambda)\left[\mathcal{G}(x,t,\lambda)\right]^{-1}.\]
\end{prop}
The domain of dressed resolvent operator is determined by the domains of bare resolvent, dressing factor and
its inverse. Clearly, dressing factor and its inverse may introduce singularities to $R^{(1)}(\lambda)$, which coincide
with discrete eigenvalues of $L^{(1)}(\lambda)$. 
\begin{corol}
The discrete eigenvalues of dressed scattering operator belong to either the discrete spectrum of bare scattering
operator or the set of singularities of dressing factor or the set of singularities of its inverse.
\end{corol}

For the gauge transform (\ref{dres}) to be of some use in finding special solutions, we shall require that the
dressed matrix function $\Psi^{(1)}$ solves the pair of linear problems
\begin{eqnarray}
L^{(1)}(\lambda)\Psi^{(1)}(x,t,\lambda) &=& 0,\label{dresaux1}\\
A^{(1)}(\lambda)\Psi^{(1)}(x,t,\lambda) &=& \Psi^{(1)}(x,t,\lambda)C(\lambda)\label{dresaux2}
\end{eqnarray} 
for $L^{(1)}(\lambda)$ and $A^{(1)}(\lambda)$ being of the form
\begin{eqnarray}
L^{(1)}(\lambda) &=& \rmi\partial_x + U^{(1)}(x,t,\lambda),\qquad U^{(1)}(x,t,\lambda) = \sum^{N}_{j=0}U^{(1)}_j(x,t) \lambda^j, \label{lax11}\\
A^{(1)}(\lambda) &=& \rmi\partial_t + V^{(1)}(x,t,\lambda),\qquad V^{(1)}(x,t,\lambda) = \sum^{M}_{k=0}V^{(1)}_k(x,t)\lambda^k. \label{lax12}
\end{eqnarray}
$U^{(1)}_j$, $j=0,\ldots, N$ and $V^{(1)}_k$, $k=0,\ldots,M$ are unknown smooth functions defined almost everywhere in $\bbr^2$ and taking values in $\mathfrak{sl}(n,\bbc)$. Similarly to the case of bare linear operators, we have that
\[\left[L^{(1)}(\lambda), A^{(1)}(\lambda)\right] = 0.\]
Clearly, the compatibility of dressed linear problems means that $U^{(1)}_j$, $j=0,\ldots, N$ are solutions to the
same NEE as $U^{(0)}_j$, $j=0,\ldots, N$.

For the linear problems (\ref{bareaux1}), (\ref{bareaux2}), (\ref{dresaux1}) and (\ref{dresaux2}) to hold simultaneously,
$\mathcal{G}$ must obey certain conditions.
\begin{prop}
Dressing factor satisfies the linear partial differential equations 
\begin{eqnarray}
&&\rmi\partial_x\mathcal{G} + U^{(1)}\mathcal{G} - \mathcal{G}U^{(0)} = 0,\label{facpde1}\\
&&\rmi\partial_t\mathcal{G} + V^{(1)}\mathcal{G} - \mathcal{G}V^{(0)} = 0.\label{facpde2}
\end{eqnarray}
\end{prop}
\begin{proof}
It is straightforward from equations (\ref{bareaux1}), (\ref{bareaux2}), (\ref{dresaux1}), (\ref{dresaux2}),
the form Lax operators (\ref{lax01}), (\ref{lax02}), (\ref{lax11}), (\ref{lax12}) and (\ref{dres}).
\end{proof}
Equations (\ref{facpde1}) and (\ref{facpde2}) do not uniquely determine dressing factor --- one should specify the
$\lambda$-dependence in $\mathcal{G}$. After picking up an appropriate ansatz for dressing factor, (\ref{facpde1})
and (\ref{facpde2}) allow one to find new solutions to the NEE from known ones.

Further we shall demonstrate how dressing technique can be applied to linear bundles in pole gauge, i.e.,
we shall restrict ourselves with the case when
\begin{equation}
\begin{split}
U^{(\sigma)}(x,t,\lambda) & = - \lambda S^{(\sigma)}(x,t),\qquad \left[S^{(\sigma)}\right]^3 = S^{(\sigma)},
\qquad \sigma  = 0,1 \ , \\
V^{(\sigma)}(x,t,\lambda) & = \sum^{M}_{k=1}V^{(\sigma)}_{k}(x,t)\lambda^k\, .
\end{split}
\label{linbunpol}
\end{equation}
Moreover, we shall impose the following additional conditions on $S^{(\sigma)}$.
\begin{enumerate}
\item $S^{(0)}(x,t)$ and $S^{(1)}(x,t)$ fulfill the pseudo-Hermiticity condition:
\begin{equation}
\mathcal{E}\left[S^{(\sigma)}(x,t)\right]^{\dag}\mathcal{E} = S^{(\sigma)}(x,t).
\label{pseudoherm}
\end{equation}
Above, $\dag$ stands for Hermitian conjugation and $\mathcal{E}=\diag(\varepsilon_1,\varepsilon_2,\ldots, \varepsilon_n)$, $\varepsilon^2_j = 1$, $j=1,2,\ldots, n$.

\item $S^{(0)}(x,t)$ and $S^{(1)}(x,t)$ satisfy the boundary condition:
\begin{equation}
\lim_{x\to\pm\infty} S^{(\sigma)}(x,t) = \Sigma,\qquad \Sigma = \left(\begin{array}{ccc}
I_r & 0 & 0 \\ 0 & 0 & 0 \\ 0 & 0 & - I_r 
\end{array}\right),
\label{boundcon}\end{equation}
where $I_r$ is the unit $r\times r$-matrix.

\item For dressing factor we require that
\[\lim_{|\lambda|\to\infty}\mathcal{G}(x,t,\lambda) = \mathcal{G}_{\infty}(x,t) < \infty .\]
\end{enumerate}

Then it turns out $S^{(0)}(x,t)$ and $S^{(1)}(x,t)$ are interrelated in a rather simple way, as stated below.
\begin{prop}
The coefficients $S^{(0)}(x,t)$ and $S^{(1)}(x,t)$ of the Lax operators $L^{(0)}(\lambda)$ and $L^{(1)}(\lambda)$ obey
the equation
\begin{equation}
S^{(1)}(x,t) = \mathcal{G}_{\infty}(x,t)S^{(0)}(x,t)\left[\mathcal{G}_{\infty}(x,t)\right]^{-1}.
\label{s1s0rel}\end{equation}
\end{prop}

\begin{proof}
After substituting the first line in (\ref{linbunpol}) into (\ref{facpde1}), we obtain the equation
\begin{equation}
\rmi\partial_x\mathcal{G} - \lambda S^{(1)}\mathcal{G} + \lambda\mathcal{G}S^{(0)} = 0,
\label{facpde1a}
\end{equation}
that can be rewritten as
\begin{equation}
\lambda S^{(1)} = \rmi\partial_x\mathcal{G}\mathcal{G}^{-1} + \lambda\mathcal{G}S^{(0)}\mathcal{G}^{-1}.
\label{facpde1b}
\end{equation}
Now we divide (\ref{facpde1b}) by $\lambda$ and set $|\lambda|\to\infty$ in it to immediately get (\ref{s1s0rel}). 
\end{proof}

The relation (\ref{s1s0rel}) shows that if we know some solution to the generalized Heisenberg ferromagnet equation (\ref{ghf})
and the asymptotic behavior of dressing factor one can find another solution. Further we shall use a dressing factor chosen in the form\footnote{We refer the reader to
\cite{book} where they can find an explanation on why we pick up dressing factor in some form or another.}
\begin{equation}
\mathcal{G}(x,t,\lambda) = I + \sum_{\gamma = 1}^{p}\frac{\lambda B_{\gamma}(x,t)}{\mu_{\gamma}(\lambda - \mu_{\gamma})},
\qquad \mu_{\gamma}\in\bbc.
\label{gansatz}
\end{equation}
We also need to know the inverse of $\mathcal{G}(x,t,\lambda)$. In order to find the latter, we make the following simple but
important observation.
\begin{prop}
If $\Psi^{(0)}(x,t,\lambda)$ is a matrix solution of bare problems, then
\begin{equation}
\tilde{\Psi}^{(0)}(x,t,\lambda) = \mathcal{E}\left\{\left[\Psi^{(0)}(x,t,\lambda^*)\right]^{\dag}\right\}^{-1}\mathcal{E}, 
\label{fs0_sym}
\end{equation}
where $*$ denotes complex conjugation, is another matrix solution to (\ref{bareaux1}) and (\ref{bareaux2}). Similarly, for
an arbitrary  matrix solution $\Psi^{(1)}(x,t,\lambda)$ of dressed problems 
\begin{equation}
\tilde{\Psi}^{(1)}(x,t,\lambda) = \mathcal{E}\left\{\left[\Psi^{(1)}(x,t,\lambda^*)\right]^{\dag}\right\}^{-1}\mathcal{E} 
\label{fs1_sym}
\end{equation} 
is again a matrix solution to (\ref{dresaux1}) and (\ref{dresaux2}).
\end{prop}
\begin{proof}
The validity of this statement directly follows from the pseudo-Hermiticity condition (\ref{pseudoherm}).
\end{proof}

The transformation introduced in (\ref{fs0_sym}) (or (\ref{fs1_sym})) can be viewed as an action of the group $\bbz_2$
onto the set of bare (or dressed) fundamental solutions $\left\{\Psi^{(0)}\right\}$ (or $\left\{\Psi^{(1)}\right\}$).
Let us denote this action by $\mathcal{K}$ and impose the natural requirement that the diagram 
\begin{equation}
\begin{CD}
\left\{\Psi^{(0)}\right\} @>\mathcal{G}>> 	\left\{\Psi^{(1)}\right\} \\
@V\mathcal{K}VV @VV\mathcal{K}V \\
\left\{\Psi^{(0)}\right\} @>>\mathcal{G}> 	\left\{\Psi^{(1)}\right\} 
\end{CD}
\label{cd}
\end{equation}
is commutative. The commutativity of (\ref{cd}) is equivalent to the symmetry condition
\begin{equation}
\mathcal{E}\left[\mathcal{G}^{\dag}(x,t,\lambda^*)\right]^{-1}\mathcal{E} = \mathcal{G}(x,t,\lambda),
\label{facsym}
\end{equation}
satisfied by dressing factor. In view of (\ref{gansatz}) and (\ref{facsym}) the inverse of dressing factor reads:
\begin{equation}
\left[\mathcal{G}(x,t,\lambda)\right]^{-1} = I + \sum_{\gamma = 1}^{p}\frac{\lambda\mathcal{E}B^{\dag}_{\gamma}(x,t)\mathcal{E}}{\mu^*_{\gamma}(\lambda - \mu^*_{\gamma})}\cdot
\label{g_inv}\end{equation}

Further considerations depend on whether or not the poles of dressing factor are real. First, we shall assume
that all of the poles are in generic position, i.e., we have that $\mu_{\gamma}\notin\bbr$, $\gamma =1,2,\ldots, p$.
For simplicity we shall impose even the stronger requirement: $\mu^*_{\gamma}\neq \mu_{\tau}$, $\gamma, \tau = 1,2,\ldots, p$.

Let us take a look at the identity $\mathcal{G}\mathcal{G}^{-1} =I$. After evaluating the residue at $\mu^*_{\tau}$, we
derive 
\begin{equation}
\left[I + \sum^{p}_{\gamma = 1}\frac{\mu^*_{\tau}B_{\gamma}(x,t)}{\mu_{\gamma}(\mu^*_{\tau} - \mu_{\gamma})}\right]\mathcal{E}B^{\dag}_{\tau}(x,t)\mathcal{E} = 0.
\label{algrel1}\end{equation}
This algebraic relation implies that $B_{\gamma}(x,t)$ is a singular matrix, therefore it admits the representation
\begin{equation}
B_{\gamma}(x,t) = X_{\gamma}(x,t)F^T_{\gamma}(x,t),
\label{BXF}\end{equation}
where $X_{\gamma}(x,t)$ and $F_{\gamma}(x,t)$ are some $n\times s$-matrices of rank $s$ ($s < n$) and superscript $T$
stands for matrix transposition. Substituting (\ref{BXF})
into (\ref{algrel1}), we get
\begin{equation}
\left[I + \sum^{p}_{\gamma = 1}\frac{\mu^*_{\tau}X_{\gamma}(x,t)F^T_{\gamma}(x,t)}{\mu_{\gamma}(\mu^*_{\tau} - \mu_{\gamma})}\right]\mathcal{E}F^{*}_{\tau}(x,t) = 0.
\label{algrel2}
\end{equation}
The following statement is straightforward from (\ref{algrel2}).
\begin{prop}
The rectangular matrices $X_{\tau}$, $\tau = 1,2,\ldots,p$ are solutions to the linear system of matrix equations
\begin{equation}
\mathcal{E}F^*_{\tau}(x,t) = \sum^{p}_{\gamma = 1}X_{\gamma}(x,t)D_{\gamma\tau}(x,t) \ ,
\label{algrel3}
\end{equation}
where 
\[D_{\gamma\tau}(x,t) = \frac{\mu^*_{\tau}F^T_{\gamma}(x,t)\mathcal{E}F^*_{\tau}(x,t)}{\mu_{\gamma}(\mu_{\gamma} - \mu^*_{\tau})}
\ \cdot\]
\label{algsys_FX}\end{prop}
The linear system (\ref{algrel3}) allows one to find $X_{\tau}$ if $F_{\tau}$ are known.
In the simplest case when $\mathcal{G}$ has a single complex pole $\mu$ (i.e., we have that $p=1$) and $X(x,t)$ and $F(x,t)$
are column vectors ($s = 1$) the matrix equation (\ref{algrel3}) has the solution
\begin{equation}
X(x,t) = \frac{\mu(\mu - \mu^*)\mathcal{E}F^*(x,t)}{\mu^* F^T(x,t)\mathcal{E}F^*(x,t)}\ \cdot
\label{XF_vec}\end{equation}
Next proposition answers the question how one can find $F_{\tau}$.

\begin{prop}
The rectangular matrices $F_{\tau}(x,t)$, $\tau = 1,2,\ldots, p$ can be expressed through any fundamental
solution $\Psi^{(0)}$ to (\ref{bareaux1}) and (\ref{bareaux2}) defined in the vicinity of the pole $\mu_{\tau}$ through the equality:	
\begin{equation}
F^T_{\tau}(x,t) = F^T_{\tau,0}(t)\left[\Psi^{(0)}(x,t,\mu_{\tau})\right]^{-1}.
\label{F_psi0}\end{equation}
Above, $F_{\tau,0}(t)$ are some yet undetermined $n\times s$-matrices. 
\label{prop_F}\end{prop}

\begin{proof}
Let us evaluate the residue in equation (\ref{facpde1b}) at an arbitrary pole $\mu_{\tau}$. The result reads:
\begin{eqnarray*}
0 & = &\lim_{\lambda\to\mu_{\tau}}(\lambda - \mu_{\tau})\left[\rmi\partial_x \mathcal{G} + \lambda\mathcal{G}S^{(0)}\right]\mathcal{G}^{-1}\\
& = &\left[\rmi\partial_xB_{\tau} + \mu_{\tau} B_{\tau} S^{(0)}\right]\left[I + \sum_{\gamma}\frac{\mu_{\tau} \mathcal{E} B^{\dag}_{\gamma}\mathcal{E}}{\mu^*_{\gamma}(\mu_{\tau} - \mu^*_{\gamma})}\right].
\end{eqnarray*}
In view of the decomposition (\ref{BXF}) and the algebraic relation (\ref{algrel2}), the equation above yields to
\begin{eqnarray*}
&& X_{\tau}\left[\rmi\partial_xF^T_{\tau} + \mu_{\tau} F^T_{\tau}S^{(0)}\right]\left[I + \sum_{\gamma}\frac{\mu_{\tau} \mathcal{E}F^*_{\gamma} X^{\dag}_{\gamma}\mathcal{E}}{\mu^*_{\gamma}(\mu_{\tau} - \mu^*_{\gamma})}\right]
=0\\
&& \Rightarrow\qquad \rmi\partial_xF^T_{\tau} + \mu_{\tau} F^T_{\tau}S^{(0)} = \Gamma F^T_{\tau},
\end{eqnarray*}
where $\Gamma(x,t)$ is a $s\times s$-matrix. Taking into account the existing ambiguity in the decomposition (\ref{BXF}),
we can set $\Gamma = 0$ without any loss of generality. Thus, we have that
\[\rmi\partial_xF^T_{\tau} + \mu_{\tau} F^T_{\tau}S^{(0)} = 0, \]
which is integrated to give (\ref{F_psi0}).
\end{proof}

The time evolution of dressed solution is driven by the dispersion law of the NEE. We recall that the dispersion law of
a completely integrable NEE whose Lax pair is in the form (\ref{lax01}) and (\ref{lax02}) (or in the form of (\ref{lax11}) and (\ref{lax12})) can be defined through the equality
\[f(\lambda) = \lim_{x\to \pm\infty}g^{(\sigma)}(x,t)\sum^{M}_{k=1}V^{(\sigma)}_{k}(x,t)\lambda^{k}\left[g^{(\sigma)}(x,t)\right]^{-1},
\qquad \sigma = 0,1,\]
where $g^{(\sigma)}(x,t)$ is a gauge transformation that casts $S^{(\sigma)}(x,t)$ into a diagonal form. The existence of $g^{(\sigma)}(x,t)$ follows from the algebraic constraint in the first line of (\ref{linbunpol}), see \cite{bgsiam} for more details. For instance, the dispersion law of the equation (\ref{ghf}) is given by
\begin{equation}
f(\lambda) = \lambda^2 \left[a \Sigma + b\left(\Sigma^2 - \frac{2r}{n}I\right)\right],
\label{ghf_disp}\end{equation}
where the matrix $\Sigma$ is introduced in (\ref{boundcon}) and $2r = \tr \left[S^{(\sigma)}\right]^2$.

\begin{prop}
The time dependence of dressed solution can fully be recovered by using the correspondence
\begin{equation}
F^T_{\tau,0} \to F^T_{\tau, 0}\rme^{-\rmi f(\mu_{\tau})t},\qquad \tau = 1,2,\ldots, p ,
\label{F0_cor}\end{equation}
where $f(\lambda)$ is the dispersion law of the NEE.
\label{trecov1}\end{prop}

\begin{proof}
The scheme of proof is similar to that of Proposition \ref{prop_F}. After evaluating the residue
in the equation 
\begin{equation}
\begin{split}
&\rmi\partial_t\mathcal{G} + \sum^{M}_{k = 1}V^{(1)}_k\lambda^k\mathcal{G} 
- \mathcal{G}\sum^{M}_{k = 1}V^{(0)}_k\lambda^k = 0 \qquad\Leftrightarrow \\
& \sum^{M}_{k = 1}V^{(1)}_k\lambda^k = - \rmi\partial_t\mathcal{G}\mathcal{G}^{-1}
 + \mathcal{G}\sum^{M}_{k = 1}V^{(0)}_k\lambda^k \mathcal{G}^{-1}
\end{split}
\label{facpde2a}
\end{equation}
at an arbitrary pole $\mu_{\tau}$, we derive
\begin{eqnarray*}
0 & = & \lim_{\lambda\to\mu_{\tau}}(\lambda - \mu_{\tau})\left[\rmi\partial_t \mathcal{G} - \mathcal{G}\sum^{M}_{k = 1}V^{(0)}_k\lambda^k \right]\mathcal{G}^{-1}\\
& = &\left[\rmi\partial_tB_{\tau} - B_{\tau} \sum^{M}_{k = 1}V^{(0)}_k\mu^k_{\tau}\right]\left[I + \sum_{\gamma}\frac{\mu_{\tau} \mathcal{E} B^{\dag}_{\gamma}\mathcal{E}}{\mu^*_{\gamma}(\mu_{\tau} - \mu^*_{\gamma})}\right]\\
& \stackrel{(\ref{BXF}), (\ref{algrel2})}{=} & X_{\tau}\left[\rmi\partial_tF^T_{\tau} - F^T_{\tau} \sum^{M}_{k = 1}V^{(0)}_k\mu^k_{\tau}\right]\left[I + \sum_{\gamma}\frac{\mu_{\tau} \mathcal{E} B^{\dag}_{\gamma}\mathcal{E}}{\mu^*_{\gamma}(\mu_{\tau} - \mu^*_{\gamma})}\right].
\end{eqnarray*}
Making use of the same argument as in the proof of Proposition \ref{prop_F}, we get 
\begin{equation}
\rmi\partial_tF^T_{\tau} - F^T_{\tau} \sum_{k}V^{(0)}_k\mu^k_{\tau} = 0.
\label{FV0}\end{equation} 
Now, we substitute (\ref{F_psi0}) into (\ref{FV0}) to obtain the linear equation
\[\rmi\partial_tF^T_{\tau, 0} - F^T_{\tau, 0} f(\mu_{\tau}) = 0,\]
which immediately leads to the correspondence (\ref{F0_cor}).
\end{proof}

Let us now consider the case when dressing factor has real poles, i.e., we have that 
\begin{equation}
\mathcal{G}(x,t,\lambda) = I + \lambda\sum^{p}_{\gamma = 1}\frac{B_{\gamma}(x,t)}{\mu_{\gamma}(\lambda - \mu_{\gamma})}\, ,\qquad \mu_{\gamma}\in\bbr\backslash\{0\}
\label{gansatz_r}\end{equation}
and its inverse
\begin{equation}
\left[\mathcal{G}(x,t,\lambda)\right]^{-1} = I + \lambda\sum^{p}_{\gamma = 1}\frac{\mathcal{E}B^{\dag}_{\gamma}(x,t)\mathcal{E}}
{\mu_{\gamma}(\lambda - \mu_{\gamma})}
\label{g_inv_r}
\end{equation}
share the same poles. In this case the identity $\mathcal{G}\mathcal{G}^{-1} = I$ produces more complicated algebraic
relations.

\begin{prop}
The residues of the dressing factor (\ref{gansatz_r}) obey the algebraic equations
\begin{eqnarray}
&& B_{\tau}\mathcal{E}B^{\dag}_{\tau} = 0, \qquad \tau = 1,2,\ldots, p ,\label{algrel1_r}\\ 
&& \Omega_{\tau} \mathcal{E}B^{\dag}_{\tau} + B_{\tau}\mathcal{E}\Omega^{\dag}_{\tau} = 0\, ,
\label{algrel2_r}
\end{eqnarray}
where 
\[\Omega_{\tau} = I + \frac{B_{\tau}}{\mu_{\tau}} + \sum_{\gamma\neq \tau}\frac{\mu_{\tau}B_{\gamma}}{\mu_{\gamma}(\mu_{\tau} - \mu_{\gamma})}\, \cdot \]
\end{prop}
\begin{proof}
The proof is analogous to the proof of Proposition \ref{algsys_FX}, e.g. to derive (\ref{algrel1_r}), one has to
evaluate the matrix coefficients before $\left(\lambda - \mu_{\tau}\right)^{-2}$ in the identity $\mathcal{G}\mathcal{G}^{-1} = I$.
Similarly, the residues of $\mathcal{G}\mathcal{G}^{-1} = I$ give rise to (\ref{algrel2_r}).
\end{proof}
From the relation (\ref{algrel1_r}) we deduce that $\det B_{\gamma}(x,t) = 0$, i.e., all of the residues of dressing factor are 
singular matrices. Therefore the decomposition (\ref{BXF}) holds true again.

\begin{corol}
In terms of the rectangular matrices $X_{\tau}(x,t)$ and $F_{\tau}(x,t)$, $\tau = 1,2,\ldots, p $ in (\ref{BXF}), the algebraic relations (\ref{algrel1_r}) and (\ref{algrel2_r}) are reduced to:
\begin{eqnarray}
&& F^T_{\tau}\mathcal{E}F^{*}_{\tau} = 0, \label{fquad}\\ 
&& \Omega_{\tau} \mathcal{E} F^*_{\tau} = X_{\tau}\alpha_{\tau},
\label{algrel2_ra}
\end{eqnarray}
where $\alpha_{\tau}(x,t)$ is a skew-Hermitian $s\times s$-matrix.
\end{corol} 
\begin{proof}
The proof is straightforward from (\ref{algrel1_r}) and (\ref{algrel2_r}).
\end{proof}

\begin{rem}
It is clear that for $\mathcal{E} = I$ (\ref{fquad}) leads to the trivial result $F_{\tau} = 0$ and $\mathcal{G}$ becomes
equal to $I$. In order to obtain a non-trivial dressing procedure we shall require that $\mathcal{E} \neq I$.
\end{rem}

The relation (\ref{algrel2_ra}) is viewed as a linear system for $X_{\tau}$. In the simplest case when $\mathcal{G}$ has a single
pole $\mu\in\bbr$, $X(x,t)$, $F(x,t)$ are column vectors and $\alpha(x,t)$ is a scalar, the equation (\ref{algrel2_ra}) admits the following simple solution
\begin{equation}
X(x,t) = \frac{\mathcal{E}F^*(x,t)}{\alpha(x,t)}\, \cdot
\label{XF_vec2}
\end{equation} 

In order to find $F_{\tau}$ and $\alpha_{\tau}$ we analyze the equation (\ref{facpde1a}) as we did in the generic case. 

\begin{prop}
The matrices $F_{\tau}(x,t)$ and $\alpha_{\tau}(x,t)$, $\tau = 1,2,\ldots, p$ can be expressed through any bare
fundamental solution $\Psi^{(0)}$ and its $\lambda$-derivative defined in the vicinity of the pole $\mu_{\tau}$ through the relations:	
\begin{eqnarray}
&& F^T_{\tau}(x,t) = F^T_{\tau,0}(t)\left[\Psi^{(0)}(x,t,\mu_{\tau})\right]^{-1},\label{F_psi0_a}\\
&& \alpha_{\tau}(x,t) = \alpha_{\tau,0}(t) - F^T_{\tau,0}(t)\left[\Psi^{(0)}(x,t,\mu_{\tau})\right]^{-1} \partial_{\lambda}\Psi^{(0)}(x,t,\mu_{\tau}) K_{\Psi^{(0)}}(\mu_{\tau})\mathcal{E}F^*_{\tau,0}(t).	
\label{alpha}
\end{eqnarray}
Above, $F_{\tau,0}(t)$ and $\alpha_{\tau,0}(t)$ are constants of integration while the matrix $K_{\Psi^{(0)}}(\lambda)$
is introduced as follows:
\[K_{\Psi^{(0)}}(\lambda) = \left[\Psi^{(0)}(x,t,\lambda)\right]^{-1}\mathcal{E}\left\{\left[\Psi^{(1)}(x,t,\lambda^*)\right]^{\dag}\right\}^{-1}\mathcal{E}.\]
\label{prop_Falpha}
\end{prop}

\begin{proof}
First, we shall derive (\ref{F_psi0_a}). For that purpose we evaluate the matrix coefficients before
$(\lambda - \mu_{\tau})^{-2}$ in the equation (\ref{facpde1b}) as given below:
\begin{eqnarray*}
0 & = &\lim_{\lambda\to\mu_{\tau}}(\lambda - \mu_{\tau})^2\left[\rmi\partial_x \mathcal{G} + \lambda\mathcal{G}S^{(0)}\right]\mathcal{G}^{-1}
 = \left[\rmi\partial_xB_{\tau} + \mu_{\tau} B_{\tau} S^{(0)}\right]\mathcal{E} B^{\dag}_{\tau}\mathcal{E}\\
&\stackrel{(\ref{BXF}), (\ref{fquad})}{=}& X_{\tau}\left[\rmi\partial_xF^T_{\tau} + \mu_{\tau} F^T_{\tau}S^{(0)}\right]  \mathcal{E}F^*_{\tau} X^{\dag}_{\tau}\mathcal{E}\qquad \Rightarrow\qquad \rmi\partial_xF^T_{\tau} + \mu_{\tau} F^T_{\tau}S^{(0)} = 0 .
\end{eqnarray*}
The latter equation gives exactly (\ref{F_psi0_a}) after an elementary integration.

Similarly, evaluating the residues in (\ref{facpde1b}) at $\mu_{\tau}$ we have
\begin{eqnarray*}
0 & = & \lim_{\lambda\to\mu_{\tau}}\frac{\partial}{\partial \lambda}\left[(\lambda - \mu_{\tau})^2\left(\rmi\partial_x \mathcal{G} + \lambda\mathcal{G}S^{(0)}\right)\mathcal{G}^{-1}\right] \\
& = & \left(\rmi\partial_xB_{\tau} + \mu_{\tau}B_{\tau} S^{(0)}\right)\mathcal{E}\Omega^{\dag}_{\tau}\mathcal{E}
 + \left(\rmi\partial_x\Omega_{\tau} + \mu_{\tau}\Omega_{\tau}S^{(0)}\right)\mathcal{E}B^{\dag}_{\tau}\mathcal{E}
+ B_{\tau}S^{(0)}\mathcal{E}B^{\dag}_{\tau}\mathcal{E} \\
& \stackrel{(\ref{BXF})}{=} & \left(\rmi\partial_xX_{\tau} F^T_{\tau}  + \rmi X_{\tau} \partial_xF^T_{\tau} + \mu_{\tau} X_{\tau}F^T_{\tau}S^{(0)}\right)
\mathcal{E}\Omega^{\dag}_{\tau}\mathcal{E} + \rmi\partial_x(\Omega_{\tau} \mathcal{E}F^*_{\tau})X^{\dag}_{\tau}\mathcal{E}\\
& - & \rmi\Omega_{\tau}\partial_x(\mathcal{E}F^*_{\tau})X^{\dag}_{\tau}\mathcal{E}
+ \mu_{\tau}\Omega_{\tau}S^{(0)}\mathcal{E}F^*_{\tau}X^{\dag}_{\tau}\mathcal{E} + X_{\tau}F^T_{\tau}S^{(0)}\mathcal{E}F^T_{\tau}X^{\dag}_{\tau}\mathcal{E}\\
& \stackrel{(\ref{algrel2_ra}), (\ref{F_psi0_a})}{=} & X_{\tau}\left(\rmi\partial_x\alpha_{\tau} + F^T_{\tau}S^{(0)}\mathcal{E}F^*_{\tau}\right)X^{\dag}_{\tau}\mathcal{E} \qquad\Rightarrow\qquad \rmi\partial_x\alpha_{\tau} + F^T_{\tau}S^{(0)}\mathcal{E}F^*_{\tau} = 0.
\end{eqnarray*}
The latter equation is integrated and gives (\ref{alpha}).
\end{proof}

What remains is to obtain the $t$-dependence of $F_{\tau,0}$ and $\alpha_{\tau,0}$. A proposition that is analogous to
Proposition \ref{trecov1} holds. 
\begin{prop}
The time dependence of dressed solution associated with a real pole dressing factor is recovered in accordance with 
the rules:
\begin{eqnarray}
F^T_{\tau,0} &\to & F^T_{\tau, 0}\,\rme^{-\rmi f(\mu_{\tau})t},\qquad \tau = 1,2,\ldots, p \ ,\label{F0_cor_a}\\
\alpha_{\tau,0} &\to & \alpha_{\tau,0} - \rmi F^T_{\tau,0}(t)\left. 
\frac{\rmd f(\lambda)}{\rmd\lambda}\right|_{\lambda = \mu_{\tau}}K_{\Psi^{(0)}}(\mu_{\tau}) \mathcal{E}F^*_{\tau,0}(t)t\ .
\label{alpha0_cor}
\end{eqnarray}
\end{prop}

\begin{proof}
Like in the proof of the Proposition \ref{trecov1}, we evaluate the matrix coefficients before $(\lambda - \mu_{\tau})^{-2}$
in the latter equation in (\ref{facpde2a}) as follows:
\begin{eqnarray*}
& 0 = &\lim_{\lambda\to\mu_{\tau}}(\lambda - \mu_{\tau})^2\left[-\rmi\partial_t \mathcal{G} + \mathcal{G}\sum^{M}_{k = 1} V^{(0)}_k  \lambda^{k} \right]\mathcal{G}^{-1} = \left[-\rmi\partial_tB_{\tau} + B_{\tau}\sum^{M}_{k = 1} V^{(0)}_k  \mu^{k}_{\tau}  \right]\mathcal{E} B^{\dag}_{\tau}\mathcal{E}\\
& \stackrel{(\ref{BXF}), (\ref{fquad})}{=} & X_{\tau}\left[-\rmi\partial_tF^T_{\tau} + F^T_{\tau}\sum^{M}_{k = 1} V^{(0)}_k  \mu^{k}_{\tau}\right]\mathcal{E}F^*_{\tau} X^{\dag}_{\tau}\mathcal{E}\quad\Rightarrow\quad \rmi\partial_tF^T_{\tau} - F^T_{\tau} \sum_{k}V^{(0)}_k\mu^k_{\tau} = 0.
\end{eqnarray*}
After taking into account (\ref{F_psi0_a}), we reduce the above equation to
\begin{equation}
\rmi\partial_tF^T_{\tau, 0} - F^T_{\tau, 0} f(\mu_{\tau}) = 0,
\label{Ftau0_de}\end{equation}
which immediately leads to an exponential dependence of $F_{\tau,0}$, namely to the rule (\ref{F0_cor_a}).

In a similar fashion, we calculate the matrix residues in (\ref{facpde2a}), namely:
\begin{eqnarray*}
0 & = & \lim_{\lambda\to\mu_{\tau}}\frac{\partial}{\partial \lambda}\left\{(\lambda - \mu_{\tau})^2\left[-\rmi\partial_t \mathcal{G} + \mathcal{G}\sum^{M}_{k = 1} V^{(0)}_k  \lambda^{k} \right]\mathcal{G}^{-1}\right\}\\
& = &\left[-\rmi\partial_tB_{\tau} + B_{\tau} \sum^{M}_{k = 1} V^{(0)}_k  \mu^{k}_{\tau} \right]\mathcal{E}\Omega^{\dag}_{\tau}\mathcal{E} + \left[-\rmi\partial_t\Omega_{\tau} + \Omega_{\tau}\sum^{M}_{k = 1} V^{(0)}_k  \mu^{k}_{\tau}\right]\mathcal{E}B^{\dag}_{\tau}\mathcal{E}\\
& + & B_{\tau}\sum^{M}_{k = 1} k V^{(0)}_k  \mu^{k-1}_{\tau}\mathcal{E}B^{\dag}_{\tau}\mathcal{E}
\stackrel{(\ref{BXF}), (\ref{algrel2_ra}), (\ref{Ftau0_de})}{=} - X_{\tau}\left[\rmi\partial_t\alpha_{\tau} - F^T_{\tau}\sum^{M}_{k = 1} kV^{(0)}_k  \mu^{k-1}_{\tau}\mathcal{E}F^*_{\tau}\right]X^{\dag}_{\tau}\mathcal{E}.
\end{eqnarray*}
Hence we derive the equation
\begin{equation}
\rmi\partial_t\alpha_{\tau} - F^T_{\tau}\sum^{M}_{k = 1} kV^{(0)}_k \mu^{k-1}_{\tau}\mathcal{E}F^*_{\tau} = 0 .
\label{alpha_pde}
\end{equation}
After substituting (\ref{F_psi0_a}), (\ref{alpha}) and (\ref{Ftau0_de}) into (\ref{alpha_pde}) and then perform simple but
somewhat tedious calculations, we get
\[\rmi\partial_t\alpha_{\tau,0} - F^T_{\tau,0}\left. \frac{\rmd f(\lambda)}{\rmd\lambda}\right|_{\lambda = \mu_{\tau}}K_{\Psi^{(0)}}(\mu_{\tau})\mathcal{E}F^*_{\tau,0} = 0.\]
This equation is integrated to give rise to the correspondence (\ref{alpha0_cor}).
\end{proof}

\section{Special Solutions}\label{solutions} 

In this section we shall present our main results, namely construction of special solutions to the generalized Heisenberg
ferromagnet equation (\ref{ghf}). The construction will be based on the procedure discussed in the previous section.
As we have already mentioned, the calculations crucially depend on the location of dressing factor poles. This is why there are two essentially different cases: soliton-like solutions that correspond to poles in generic position and quasi-rational solutions associated with real poles. We shall start with solutions of the soliton type.

\subsection{Soliton Solutions}
\label{solitons}

In this subsection we shall require that dressing factor has simple poles with non-zero imaginary parts. In fact, we shall
be interested in a dressing factor with a single pole $\mu = \omega + \rmi \kappa$. Without loss of generality
we can assume that $\omega >0$ and $\kappa>0$. In view of (\ref{g_inv}) the inverse of $\mathcal{G}$ has a single pole
$\mu^* = \omega - \rmi \kappa$. Then the following theorem holds true.

\begin{theo}
The simplest soliton-like solution to the generalized Heisenberg ferromagnet equation (\ref{ghf}) related to the Lie algebra $\mathfrak{su}(m, n-m)$, i.e., the soliton-like solution connected to a single pair of discrete eigenvalues $\omega \pm \rmi \kappa$ of the scattering operator, is given by:	
\begin{equation}
\begin{split}
S^{(1)}_{ij}(x,t) &= \sum^{r}_{k=1} \left\{\delta_{ik} + \frac{2\rmi \kappa \varepsilon_ih_{ik}(x,t)f_{ik}}{(\omega -\rmi\kappa)\left[
\rme^{-2\kappa (x - 2u_{+}t)}g_{+} + \rme^{-\frac{8\omega\kappa rb t}{n}}g_{0} + \rme^{2\kappa (x - 2u_{-}t)}g_{-}\right]}\right\}\\
& \times \left\{\delta_{kj} - \frac{2\rmi \kappa \varepsilon_kh^*_{jk}(x,t)f^*_{jk}}{(\omega +\rmi\kappa)\left[
\rme^{-2\kappa (x - 2u_{+}t)}g_{+} + \rme^{-\frac{8\omega\kappa rb t}{n}}g_{0} + \rme^{2\kappa (x - 2u_{-}t)}g_{-}\right]}\right\}\\
& - \sum^{n}_{k=n-r+1} \left\{\delta_{ik} + \frac{2\rmi \kappa \varepsilon_ih_{ik}(x,t)f_{ik}}{(\omega -\rmi\kappa)\left[
\rme^{-2\kappa (x - 2u_{+}t)}g_{+} + \rme^{-\frac{8\omega\kappa rb t}{n}}g_{0} + \rme^{2\kappa (x - 2u_{-}t)}g_{-}\right]}\right\}\\
& \times \left\{\delta_{kj} - \frac{2\rmi \kappa \varepsilon_kh^*_{jk}(x,t)f^*_{jk}}{(\omega +\rmi\kappa)\left[
\rme^{-2\kappa (x - 2u_{+}t)}g_{+} + \rme^{-\frac{8\omega\kappa rb t}{n}}g_{0} + \rme^{2\kappa (x - 2u_{-}t)}g_{-}\right]}\right\}.
\end{split}
\label{gen_ghfsol}\end{equation} 
Above, $h(x,t)$ is a Hermitian matrix defined through the relations:
\begin{equation}
\begin{split}
h_{ik}(x,t) &= \left\{ \begin{array}{l}
\rme^{-2\kappa(x - 2u_{+}t)}, \qquad i,k = 1,2,\ldots, r;\\
\rme^{-\rmi [\omega x - (\omega^2 - \kappa^2)(a+b) t]}\rme^{-\kappa\left[x - 2\omega\left(a + b\frac{n-4r}{n}\right) t\right]}, \quad i = 1,2,\ldots, r,\\ k= r+1,r+2,\ldots, n-r; \\
\rme^{-2\rmi [\omega x - a(\omega^2 - \kappa^2)t]}\rme^{\frac{4\omega\kappa b(n-2r)t}{n}}, \quad i = 1,2,\ldots, r,\\
k= n-r+1,n-r+2,\ldots, n; \\
\rme^{- \frac{8\omega\kappa rb t}{n}},\qquad  i,k = r+1,r+2,\ldots, n-r;\\
\rme^{-\rmi [\omega x - (\omega^2 - \kappa^2)(a-b) t]}\rme^{\kappa\left[x - 2\omega\left(a - b\frac{n-4r}{n}\right) t\right]},
\quad i = r+1,r+2,\ldots, n-r,\\
k= n-r+1,n-r+2,\ldots, n; \\
\rme^{2\kappa(x - 2u_{-}t)}, \qquad i,k = n-r+1,n-r+2,\ldots, n;\\
\end{array}\right.
\end{split} 
\label{hik}
\end{equation}
$f$ is another Hermitian matrix with constant elements that parametrize the solution, $g_{\pm}, g_0\in\bbr$ and 
\begin{equation}
u_{\pm} = \omega\left(a \pm b \frac{n-2r}{n}\right).
\label{upm}\end{equation}
We assume that $(g_{+}, g_{0}, g_{-}) \neq (0,0,0)$.
\label{ghf_sol}\end{theo}

\begin{proof}
In order to derive (\ref{gen_ghfsol}), we start from the bare solution
\begin{equation}
S^{(0)}(x,t) = \Sigma = \left(\begin{array}{ccc}
I_r & 0 & 0 \\ 0 & 0 & 0 \\ 0 & 0 & -I_r
\end{array}\right) ,
\label{baresol_ghf0}
\end{equation}
whose fundamental solution can be picked up to be
\begin{equation}
\Psi^{(0)}(x,t,\lambda) = \rme^{-\rmi\lambda \Sigma x}.
\label{barefas_ghf0}
\end{equation}
From this point on we shall assume that the residue $B(x,t)$ of the dressing factor (\ref{gansatz}) is a
matrix of rank $1$. In view of (\ref{F_psi0}) and (\ref{barefas_ghf0}) the vector $F(x,t)$ introduced
in the decomposition (\ref{BXF}) acquires the form
\begin{equation}
F(x,t) = \rme^{\rmi\mu\Sigma x} F_0(t),
\end{equation}
which written in components gives:
\begin{equation}
\begin{split}
F_i(x,t) &= \rme^{\rmi\mu x} F_{0,i}(t) = \rme^{\rmi\omega x}\rme^{-\kappa x}F_{0,i}(t),\qquad i=1,2,\ldots, r;\\
F_i(x,t) &=  F_{0,i}(t),\qquad i=r+1,r+2,\ldots, n-r;\\
F_i(x,t) &= \rme^{-\rmi\mu x} F_{0,i}(t) = \rme^{-\rmi\omega x}\rme^{\kappa x}F_{0,i}(t),\qquad i=n-r+1,n-r+2,\ldots,n.
\end{split}
\label{F_sol}
\end{equation}
According to (\ref{ghf_disp}) and (\ref{F0_cor}) the time evolution in (\ref{F_sol}) is recovered through the correspondences:
\begin{equation}
\begin{split}
F_{0,i} &\to F_{0,i}\exp\left[-\rmi\mu^2\left(a + \frac{n-2r}{n}b\right)t\right],\qquad i=1,2,\ldots, r ;\\
F_{0,i} &\to F_{0,i}\exp\left(\frac{2\rmi\mu^2rb}{n} t\right),\qquad i=r+1,r+2,\ldots, n-r ;\\
F_{0,i} &\to F_{0,i}\exp\left[-\rmi\mu^2\left(-a + \frac{n-2r}{n}b\right)t\right],\qquad i=n-r+1,n-r+2,\ldots, n.
\end{split}
\label{F_solcor}\end{equation}
Taking into account (\ref{BXF}) and (\ref{XF_vec}), for the matrix elements of $B(x,t)$ we have 
\begin{equation}
B_{ik}(x,t) = \frac{(\omega + \rmi\kappa)2\rmi\kappa \varepsilon_iF^{*}_i(x,t)F_k(x,t)}{(\omega - \rmi\kappa) F^T(x,t)\mathcal{E}F^{*}(x,t)}\;
\cdot
\label{B_sol}\end{equation}
After substituting (\ref{F_sol}) and (\ref{F_solcor}) into the denominator of (\ref{B_sol}), we obtain
\begin{equation}
F^T(x,t)\mathcal{E}F^*(x,t) = \rme^{-2\kappa (x - 2u_{+}t)} g_{+} + \rme^{-\frac{8\omega\kappa rb}{n}t} g_{0} + \rme^{2\kappa (x - 2u_{-}t)}g_{-} \ ,
\label{FEF}\end{equation}
where $u_{\pm}$ are the same as in (\ref{upm}) and we have introduced the notations 
\[g_{+} = \sum^r_{k=1} \varepsilon_k |F_{0,k}|^2 ,\qquad
g_{0} = \sum^{n-r}_{k=r+1}\varepsilon_k|F_{0,k}|^2, \qquad g_{-} = \sum^{n}_{k=n-r+1}\varepsilon_k|F_{0,k}|^2 \, .\]
On other hand, for the numerator of $B(x,t)$ (\ref{F_sol}) and (\ref{F_solcor}) lead to 
\begin{equation}
F^*_i(x,t)F_k(x,t) = h_{ik}(x,t)f_{ik},\qquad f_{ik} = F^*_{0,i}F_{0,k}\ ,
\label{FiFk}\end{equation}
where the functions $h_{ik}$ are defined in (\ref{hik}). After substituting (\ref{B_sol}), (\ref{FEF}) and (\ref{FiFk})
into the interrelation (\ref{s1s0rel}) between bare and dressed solutions, which can be written as
\begin{equation}
\begin{split}
S^{(1)}_{ij}(x,t) &= \sum^{r}_{k=1} \left(\delta_{ik} + \frac{B_{ik}(x,t)}{\mu}\right)\left(\delta_{kj} + \frac{\varepsilon_{k}\varepsilon_{j}B^*_{jk}(x,t)}{\mu^*}\right) \\
& - \sum^{n}_{k=n-r+1} \left(\delta_{ik} + \frac{B_{ik}(x,t)}{\mu}\right)\left(\delta_{kj} + \frac{\varepsilon_{k}\varepsilon_{j}B^*_{jk}(x,t)}{\mu^*}\right),
\end{split}
\label{s1s0rel2}\end{equation}
one finally gets (\ref{gen_ghfsol}).
\end{proof}

It is seen that the obtained solution (\ref{gen_ghfsol}) is not a traveling wave. Moreover, its denominator
may have zeros, hence and the solution may be singular. 

Since the equation (\ref{hf}) represents a reduction of the generalized Heisenberg ferromagnet equation its solutions are particular cases of those to (\ref{ghf}). Nevertheless, due to its high importance by itself we shall formulate the corresponding results for the Heisenberg ferromagnet equation as separate statements.
Thus, we have the following theorem.

\begin{theo}
The simplest soliton-like solution to the Heisenberg ferromagnet equation (\ref{hf}) related to $\mathfrak{su}(m, n-m)$ has the form:
\begin{equation}
\begin{split}
S^{(1)}_{ij}(x,t) &= 2\sum_{k=1}^{r}\left\{\delta_{ik} + \frac{2\rmi\kappa\varepsilon_i h_{ik}f_{ik}}{(\omega - \rmi \kappa)\left[g_{+} \rme^{-2\kappa (x - 4\omega t) } + g_{-}\rme^{2\kappa(x - 4\omega t)}\right]}\right\}\times \\
& \left\{\delta_{kj} - \frac{2\rmi\kappa\varepsilon_k h^*_{jk}f^*_{jk}}{(\omega + \rmi \kappa)\left[g_{+} \rme^{-2\kappa (x - 4\omega t) } + g_{-} \rme^{2\kappa(x - 4\omega t)}\right]}\right\} - \delta_{ij} \ ,
\end{split}
\label{gen_hfsol}\end{equation}
where $f_{ik}\in\bbc$, $f_{ki} = f^*_{ik}$, $i,k = 1,2,\ldots,n$, $g_{+}$ and $g_{-}$ are real numbers such that $(g_{+}, g_{-}) \neq (0,0)$, $h$ is defined through the relations:
\begin{eqnarray*}
h_{ik}(x,t) = 
\left\{ \begin{array}{cl}
\rme^{-2\kappa(x-4\omega t)}, & i,k = 1,2,\ldots r;\\
\rme^{-2\rmi [\omega x - 2(\omega^2-\kappa^2)t]}, & i = 1,2,\ldots r,\quad k= r+1,r+2,\ldots 2r; \\
\rme^{2\kappa(x-4\omega t)}, & i,k = r+1,r+2,\ldots 2r
\end{array}\right. 
\end{eqnarray*}
and the condition that $h_{ki} = h^*_{ik}$.
\label{hf_sol}
\end{theo}

\begin{proof}
The scheme of the proof coincides with that of Theorem \ref{ghf_sol}. One starts from the bare
solution to the equation (\ref{hf}) 
\begin{equation}
S^{(0)}(x,t) = \Sigma_3 = \left(\begin{array}{cc}
I_r & 0 \\ 0 & - I_r
\end{array}\right),
\label{baresol_hf}
\end{equation}
whose fundamental solution is chosen in the form
\begin{equation}
\Psi^{(0)}(x,t,\lambda) = \rme^{-\rmi\lambda \Sigma_3 x}.
\label{barefas_hf}\end{equation}
Applying a single pole dressing factor to (\ref{s1s0rel}), (\ref{XF_vec}), (\ref{F_psi0}), (\ref{F0_cor})
and taking into account that the dispersion law of (\ref{hf}) is
\begin{equation}
f(\lambda) = 2\lambda^2 \Sigma_3\ ,
\label{hf_disp}\end{equation}
one derives the solution (\ref{gen_hfsol}).
\end{proof}

Let us consider in more detail the simplest case of Heisenberg ferromagnet equation related to either $\mathfrak{su}(2)$
or $\mathfrak{su}(1,1)$, i.e., we have that $n=2$ and $r=1$. We shall require that $g_{+}g_{-}\neq 0$. Then the solution (\ref{gen_hfsol}) simplifies to
\begin{equation}
\begin{split}
S^{(1)}_{11}(x,t) &= - S^{(1)}_{22}(x,t) = 1- \frac{8\kappa^2\varepsilon_1\varepsilon_2}{(\omega^2 + \kappa^2)\left[\rme^{\xi(x,t)} + \varepsilon_1\varepsilon_2\rme^{-\xi(x,t)}\right]^2}\ ,\\ 
S^{(1)}_{12}(x,t) &= \varepsilon_1\varepsilon_2 \left[S^{(1)}_{21}(x,t)\right]^* \\
& = \frac{4\kappa}{(\omega^2 + \kappa^2)\left[\rme^{-\xi(x,t)} + \varepsilon_1\varepsilon_2\rme^{\xi(x,t)}\right]}\left[\frac{2\kappa\cos \varphi(x,t)}{1 + \varepsilon_1\varepsilon_2\rme^{2\xi(x,t)}} + \kappa \cos\varphi(x,t)\right.\\
& \left. + \omega \sin\varphi(x,t) 
- \frac{2\rmi\kappa\sin \varphi(x,t)}{1 + \varepsilon_1\varepsilon_2\rme^{2\xi(x,t)}} +\rmi \omega\cos\varphi(x,t) - \rmi \kappa \sin\varphi(x,t)\right],
\end{split}
\label{gen_hfsol1}
\end{equation}
where 
\begin{eqnarray*}
\xi(x,t) &=& 2\kappa(x - 4\omega t) + \frac{1}{2}\ln\frac{f_{22}}{f_{11}} \ ,\\
\varphi(x,t) &=& 2\left[\omega x - 2(\omega^2 - \kappa^2)t + \frac{\arg f_{21}}{2}\right].
\end{eqnarray*}
Equivalently, the soliton type solution (\ref{gen_hfsol1}) can be written down in a vector form by making use of the equalities:
\[ S^{(1)}_{11} = S^{(1)}_3,\qquad \re S^{(1)}_{12} = S^{(1)}_1,\qquad \im S^{(1)}_{12} = - S^{(1)}_2.\]
This way we obtain the following result:
\[\begin{split}
S^{(1)}_{1}(x,t) & = \frac{4\kappa \left[\frac{2\kappa\cos \varphi(x,t)}{1 + \varepsilon_1\varepsilon_2\rme^{2\xi(x,t)}}
+ \kappa \cos\varphi(x,t) + \omega \sin\varphi(x,t) \right]}{(\omega^2 + \kappa^2)\left[\rme^{-\xi(x,t)} + \varepsilon_1\varepsilon_2\rme^{\xi(x,t)}\right]}\, ,\\
S^{(1)}_{2}(x,t) & =\frac{4\kappa \left[\frac{2\kappa\sin \varphi(x,t)}{1 + \varepsilon_1\varepsilon_2\rme^{2\xi(x,t)}} - \omega\cos\varphi(x,t) + \kappa \sin\varphi(x,t)\right]}{(\omega^2 + \kappa^2)\left[\rme^{-\xi(x,t)} + \varepsilon_1\varepsilon_2\rme^{\xi(x,t)}\right]}\, ,\\
S^{(1)}_{3}(x,t) & = 1- \frac{8\kappa^2\varepsilon_1\varepsilon_2}{(\omega^2 + \kappa^2)\left[\rme^{\xi(x,t)} + \varepsilon_1\varepsilon_2\rme^{-\xi(x,t)}\right]^2}\, \cdot
\end{split}\]

Next, we shall consider the case when the poles of dressing factor are all real, i.e., the discrete eigenvalues of
the scattering operator now lie on its continuous spectrum. This degeneracy in the spectrum leads to certain
"degeneracy" in the soliton type solutions --- the complex exponential functions turn into the product of oscillating
exponents and polynomials. Thus, the soliton solutions being the ratio of such exponents become the ratio of quasi-polynomials, that is quasi-rational functions of $x$ and $t$.

\subsection{Quasi-rational Solutions}
\label{quasi-rational}

Like in the generic case we shall focus here on a dressing factor with a single pole $\mu >0$. Due to (\ref{g_inv_r})
the inverse of dressing factor has a simple pole at the same point. Then the following statement holds true.

\begin{theo}
The simplest quasi-rational solution to the generalized Heisenberg ferromagnet equation (\ref{ghf}) for the pseudo-unitary
algebra associated with a single double discrete eigenvalue $\mu >0$ is given by the $n\times n$-matrix:
\begin{equation}
\begin{split}
S^{(1)}_{ij}(x,t) = \sum_{k=1}^{r}\left[\delta_{ik} - \frac{\rmi\varepsilon_i h_{ik}(x,t)f_{ik}}{\mu g_{+} (x - 2v_{+} t) - \mu g_{-} (x - 2v_{-}t)}\right]\\
\times\left[\delta_{kj} + \frac{\rmi\varepsilon_kh^*_{jk}(x,t)f^*_{jk}}{\mu g_{+} (x - 2v_{+} t) - \mu g_{-} (x - 2v_{-}t)}\right]\\
 -\sum_{k=n-r+1}^{n}\left[\delta_{ik} - \frac{\rmi\varepsilon_i h_{ik}(x,t)f_{ik}}{\mu g_{+} (x - 2v_{+} t) - \mu g_{-}(x - 2v_{-}t)}\right]\\
\times\left[\delta_{kj} + \frac{\rmi\varepsilon_kh^*_{jk}(x,t)f^*_{jk}}{\mu g_{+} (x - 2v_{+} t) - \mu g_{-} (x - 2v_{-}t)}\right],
\end{split}
\label{quasi_ghfsol}
\end{equation}
where $h_{ik}(x,t)$ are elements of a Hermitian matrix defined through:	
\begin{equation}
h_{ik}(x,t) = \left\{\begin{array}{rl}
\rme^{-\rmi\mu(x- v_{+}t)}, & i = 1,2,\ldots, r,\quad k = r+1,r+2,\ldots,n-r;\\
\rme^{-2\rmi\mu(x - a\mu t)}, & i = 1,2,\ldots, r, \ k = n-r+1, n-r+2, \ldots, n;\\
\rme^{-\rmi\mu(x - v_{-}t)}, & i = r+1, r+2, \ldots, n-r, \\ & k = n-r+1, n-r+2, \ldots, n;\\
1, & \mbox{otherwise}.
\end{array}\right. 
\label{hik2}
\end{equation}
The constant $n\times n$-matrix $f$, $f^{\dag} = f$, $g_{\pm}\in\bbr$, $(g_{+}, g_{-}) \neq (0,0)$ and $v_{\pm} = \mu (a \pm b)$ parametrize the solution.
\label{ghf_quasi}
\end{theo}

\begin{proof}
In constructing the above quasi-rational solution we once again use (\ref{baresol_ghf0}) and (\ref{barefas_ghf0}) 
as a bare solution and a bare fundamental solution respectively. Like in the soliton case, we shall restrict ourselves
with a dressing factor whose residue is a matrix of rank $1$. Then equation (\ref{F_psi0_a}) again leads to the expression (\ref{F_sol}) for the column vector $F(x,t)$ appearing in the decomposition (\ref{BXF}). For the components of the residue of dressing factor the equation (\ref{XF_vec2}) reads:
\begin{equation}
B_{ik}(x,t) = \frac{\varepsilon_iF^*_i(x,t)F_k(x,t)}{\alpha (x,t)} \ ,\qquad i, k = 1,2,\ldots, n. 
\label{B_quasi}
\end{equation} 
The $t$-dependence of $F(x,t)$ is recovered from (\ref{F_solcor}) as in the soliton case while (\ref{ghf_disp}) and (\ref{alpha0_cor}) lead to the correspondence:
\begin{equation}
\alpha_0 \to \alpha_0 - 2\rmi\mu t F^T_0\left[a\Sigma + b \left(\Sigma^2 - \frac{2r}{n}I\right)\right]\mathcal{E} F^*_0
\label{alpha_solcor}\end{equation} 
for the scalar function $\alpha$. Taking into account (\ref{alpha}), (\ref{F_sol}), (\ref{F_solcor}) and (\ref{alpha_solcor}), for the denominator in (\ref{B_quasi}) we get
\begin{equation}
\alpha(x,t) = \rmi (x - 2v_{+} t) g_{+} - \rmi (x - 2v_{-} t) g_{-},
\label{alpha_quasi}
\end{equation}
where we have introduced the parameters
\[v_{\pm} = \mu (a \pm b), \qquad g_{+} = \sum^{r}_{k=1} \varepsilon_k |F_{0,k}|^2, \qquad g_{-} = \sum^{n}_{k=n-r+1} \varepsilon_k |F_{0,k}|^2.\]
On other hand, (\ref{F_sol}) and (\ref{F_solcor}) allow one to write 
\begin{equation}
F^*_i(x,t)F_k(x,t) = h_{ik}(x,t)f_{ik},\qquad f_{ik} = F^*_{0,i}F_{0,k} ,
\label{FiFk2}\end{equation}
where the matrix elements $h_{ik}(x,t)$, $i,k =1,2,\ldots, n$ are the same as in (\ref{hik2}). After substituting (\ref{B_quasi}), (\ref{alpha_quasi}) and
(\ref{FiFk2}) into the interrelation between bare and dressed solutions in the form (\ref{s1s0rel2}), one derives (\ref{quasi_ghfsol}).
\end{proof}

Like the soliton type solutions discussed in the previous subsection, the quas-rational solution (\ref{quasi_ghfsol}) is
not a traveling wave and it may exhibit (pole) singularities. Theorem \ref{ghf_quasi} has its analogue for the Heisenberg ferromagnet equation (\ref{hf}).

\begin{theo}
The simplest quasi-rational solution to the equation (\ref{hf}) for $\mathfrak{su}(m, n-m)$ is associated with a double discrete eigenvalue $\mu >0$ and looks as follows:
\begin{equation}
\begin{split}
S^{(1)}_{ij}(x,t) & = 2\sum_{k=1}^{r}\left(\delta_{ik} - \frac{\rmi\varepsilon_i h_{ik}f_{ik}}{2g\mu (x - 4\mu t)}\right)\left(\delta_{kj}
+ \frac{\rmi\varepsilon_kh^*_{jk}f^*_{jk}}{2g\mu (x - 4\mu t)}\right) - \delta_{ij}\ ,
\end{split}
\label{quasi_hfsol}\end{equation}
where $g$ is a real number ($g\neq 0$); $f_{ik} \in\bbc$, $f_{ki} = f^*_{ik}$, $i,k = 1,2, \ldots, n$, and $h$ is a Hermitian matrix defined through:
\[h_{ik}(x,t) = \left\{\begin{array}{ll}
1, & i,k = 1,2,\ldots r ;\\
\rme^{-2\rmi\mu(x - 2\mu t)}, & i = 1,2,\ldots, r, \quad k = r+1, r+2, \ldots, 2r ;\\
1, & i,k = r+1,r+2,\ldots 2r.
\end{array}\right.\]
\end{theo}

\begin{proof}
The proof here resembles that of Theorem \ref{ghf_quasi}. We choose the bare solution of
(\ref{hf}) in the form (\ref{baresol_hf}) and the corresponding fundamental solution as in (\ref{barefas_hf}).
Making use of a single real pole dressing factor to (\ref{s1s0rel}), (\ref{XF_vec2})--(\ref{alpha0_cor}) and (\ref{hf_disp})
we get the dressed solution (\ref{quasi_hfsol}).
\end{proof}

Let us consider now the simplest case of Heisenberg ferromagnet equation for the Lie algebra 
$\mathfrak{su}(1,1)$\footnote{As we discussed in the previous section, we must require that	$\mathcal{E}\neq I$ to obtain a non-trivial result. This is why Heisenberg ferromagnet equation related to $\mathfrak{su}(2)$ is not possible here.}, i.e.,
we have that $n=2$ and $m=r=1$. Then the quasi-rational solution (\ref{quasi_hfsol}) simplifies to
\[\begin{split}
S^{(1)}_{11}(x,t) & = - S^{(1)}_{22}(x,t) = 1 + \frac{1}{2\mu^2 (x - 4\mu t)^2},\\
S^{(1)}_{12}(x,t) & = - \left[S^{(1)}_{21}(x,t)\right]^* =
\left(1 - \frac{\rmi}{2\mu (x - 4\mu t)}\right)
\frac{\rmi\rme^{-\rmi\varphi(x,t)}}{\mu (x - 4\mu t)} \, ,
\end{split}\]
where
\[\varphi(x,t) = 2\left(\mu x - 2\mu^2 t + \frac{\arg f_{21}}{2}\right).\]
Written in a vector form, the above solution reads:
\[\begin{split}	
S^{(1)}_{1}(x,t) & =  \frac{\sin\varphi(x,t)}{\mu (x - 4\mu t)}
+ \frac{\cos\varphi(x,t)}{2\mu^2 (x - 4\mu t)^2} \, ,\\
S^{(1)}_{2}(x,t) & =  \frac{\sin\varphi(x,t)}{2\mu^2 (x - 4\mu t)^2}
- \frac{\cos\varphi(x,t)}{\mu (x - 4\mu t)} \, ,\\
S^{(1)}_{3}(x,t) & = 1 + \frac{1}{2\mu^2 (x - 4\mu t)^2}\, \cdot
\end{split}\]

\section{Conclusion}

Special solutions to an integrable generalization of Heisenberg ferromagnet equation have explicitly been derived.
This generalized Heisenberg ferromagnet equation has a Lax representation related to the pseudo-unitary algebra.
The existence and the form of the Lax pair hints at applying dressing technique with dressing factors being
meromorphic functions in the spectral parameter $\lambda$. More specifically, we have used dressing factors with
simple poles only. Depending on the location of the poles, one distinguishes between two classes of constant background
solutions. The first class consists of soliton-like solutions. Those are special solutions that correspond to complex
poles with non-zero imaginary parts and tend exponentially fast to some constants as $x\to\pm\infty$. The second class
includes quasi-rational solutions which are related to real poles. Those solutions tend polynomially fast to
constants as $x\to\pm\infty$. Both classes consist of solutions that are not traveling waves. Another important feature
is that the solutions may not be defined on the whole $x$-axis but may have certain singularities.

One may extend the results in the current paper by studying generalized Heisenberg ferromagnet equations related
to other simple Lie algebras, e.g. orthogonal or symplectic Lie algebras. Since those equations could be viewed as
reductions of (\ref{ghf}), the corresponding dressing factors can be obtained from the ansatz (\ref{gansatz})
by imposing certain reduction conditions. Generally speaking, the solutions will then be associated with four discrete
eigenvalues of the scattering operator.

We have considered a generalization of Heisenberg ferromagnet equation whose Lax pair represents a linear bundle
in pole gauge. It is plausible to ask for a generalization whose Lax pair is a rational bundle that is invariant
under the transformation $\lambda\to 1\slash\lambda$. Such a NEE was already introduced in \cite{varna2020} and it can be viewed as
an integrable deformation of certain reduction of the generic equation (\ref{ghf}). How to integrate this deformed
NEE through dressing method is still an open question.

We have constructed particular solutions to (\ref{ghf}) when $S$ obeys the simplest asymptotic behavior possible.
Clearly, one may require that $S$ satisfies some other boundary conditions, e.g. (quasi-)periodic or non-trivial
background conditions. Those boundary conditions, however, lead to significantly more complicated scattering theory
for the Lax operators. All this is envisaged to be part of our future work.

\section*{Acknowledgments}

The work has been financially supported by Grant KP--06--N 62/5 of Bulgarian National Science Fund.


\begin{thebibliography}{99}


\bibitem{BoPo90}
Borovik, A. E. and Popkov, V. Yu., Completely Integrable Spin-1 Chains, {\it Sov. Phys. JETPH} {\bf 71}
(1990) 177--85.
\bibitem{forkul}
Fordy, A., Kulish, P., Nonlinear Schr\"{o}dinger Equations and Simple Lie Algebras, {\it Commun. Math. Phys.}
{\bf 89} (1983) 427--443.
\bibitem{side9}
Gerdjikov, V. S., Grahovski, G. G., Mikhailov, A. V., Valchev, T. I., Polynomial Bundles and Generalised Fourier
Transforms for Integrable Equations on A.III-type Symmetric Spaces, {\it SIGMA} {\bf 7} 096 (2011) 48 pages.
\bibitem{book}
Gerdjikov, V., Vilasi, G. and Yanovski, A., \emph{Integrable Hamiltonian Hierarchies. Spectral and
Geometric Methods}, Lecture Notes in Physics {\bf 748}, Springer, Berlin, 2008.
\bibitem{golsok}
Golubchik, I. Z., Sokolov, V. V., Multicomponent Generalization of the Hierarchy
of the Landau-Lifshitz Equation, {\it Theor. Math. Phys.} {\bf 124}, n. 1 (2000) 909--917.
\bibitem{ishim}
Ishimori, Y., Multi-vortex Solutions of a Two-dimensional Nonlinear Wave Equation,
{\it Prog. Theor. Phys.} {\bf 72} (1984) 33--37.
\bibitem{mikh2}
Mikhailov, A. V., The Reduction Problem and Inverse Scattering Method, {\it Physica D} {\bf 3} (1981) 73--117.
\bibitem{mmnl}
Myrzakulov, R., Mamyrbekova, G., Nugmanova, G., Lakshmanan, M., Integrable (2+1) Dimensional 
Spin Model: Geometric and Gauge Equivalent Counterparts, Solitons and Coherent Structures,
{\it Phys. Lett. A} {\bf 233} (1997) 391--396.
\bibitem{blue-bible}
Takhtadjan, L., Faddeev, L., \emph{The Hamiltonian Approach to Soliton Theory}, Springer Verlag, Berlin, 1987.
\bibitem{varna2020}
Valchev, T., Multicomponent Nonlinear Evolution Equations of the Heisenberg Ferromagnet Type.
Local Versus Nonlocal Reductions, In: {\it Geometry, Integrability and Quantization}, {\bf 22},
Eds.: I. Mladenov, V. Pulov and A. Yoshioka, Avangard Prima, Sofia, 2021, 274--285 (Proc. XXII-nd
International Conference "Geometry, Integrability and Quantization", June 8--13, 2020, Varna, Bulgaria).
\bibitem{bgsiam}
Valchev, T., A Generic Nonlinear Evolution Equation of Magnetic Type I. Reductions, In: {\it Advanced Computing in
Industrial Mathematics}, Studies in Computational Intelligence, {\bf 1076}, Eds: I. Georgiev, H. Kostadinov, E. Lilkova,
Springer, 2023, 166--177  (Proc. of 15th Annual Meeting of the Bulgarian Section of SIAM, December 15-17, 2020, Sofia,
Bulgaria).
\bibitem{pliska2017}
Valchev, T. I., Yanovski, A. B., New Reductions of a Matrix Generalized Heisenberg Ferromagnet Equation, {\it Pliska
Stud. Math.} {\bf 29} (2018) 179--188 (Proc. of Fourth International Conference NTADES, 18--22 June, 2017, Sofia, Bulgaria). 
\bibitem{zakh-mikh}
Zakharov, V., Mikhailov, A., On the Integrability of Classical Spinor Models in Two-dimensional Space-time,
{\it Commun. Math. Phys.} {\bf 74} (1980) 21--40.
\bibitem{ZS}
Zakharov, V. E., Shabat, A. B., A Scheme for Integrating Nonlinear Equations of Mathematical Physics by
the Method of the Inverse Scattering Transform II, {\it Funct. Anal. and Appl.} {\bf 13} (1979) 13--23.

\end{thebibliography}
\end{document}